\documentclass[a4paper, 11pt]{article}
\usepackage{appendix}

\usepackage{amsmath}
\usepackage{amssymb}
\usepackage{amsfonts}
\usepackage{latexsym}
  \usepackage{epic}
  \usepackage{eepic}
 \usepackage{epstopdf}
  \usepackage{color}
\usepackage{amsfonts}
\usepackage[shortend,vlined,boxed]{algorithm2e}
\usepackage{multirow}

\newtheorem{lemma}{Lemma}[section]

\newtheorem{theorem}[lemma]{Theorem}
\newtheorem{corollary}[lemma]{Corollary}
\newtheorem{obs}[lemma]{Observation}
\newtheorem{definition}[lemma]{Definition}
\newtheorem{claim}{Claim}
\newtheorem{proposition}[lemma]{Proposition}

\def\qed{\vrule height .9ex width .8ex depth -.1ex}

\newenvironment{proof}{\noindent {\bf Proof.}$~$\ignorespaces}
{\hspace*{\fill} \qed\medskip}

\usepackage{vmargin}
\setmarginsrb{.9in}{1in}{.9in}{1in}{0mm}{0mm}{0mm}{7mm}

\newtheorem{RedRule}{Reduction Rule}

\title{Kernel(s) for Problems With no Kernel: On Out-Trees With Many Leaves (Extended~Abstract)}

\author{Henning Fernau\thanks{Univ. Trier, FB 4�Abteilung Informatik, 54286 Trier, Germany.
\texttt{\{fernau|raible\}@uni-trier.de}} \and Fedor V. Fomin\thanks{%
Department of Informatics, University of Bergen, Bergen
Norway. \newline $~$\hspace{.75cm}\texttt{\{fedor.fomin|daniello|saket.saurabh|yngve.villanger\}@ii.uib.no}}\addtocounter{footnote}{-1} \and Daniel Lokshtanov\footnotemark
\and  \addtocounter{footnote}{-2} Daniel Raible\footnotemark \and 
 Saket Saurabh\footnotemark \and \addtocounter{footnote}{-1} Yngve Villanger\footnotemark}

\begin{document}
\date{}
\maketitle

\begin{abstract}
The {\sc $k$-Leaf Out-Branching} problem is to find an out-branching
(i.e. a rooted oriented spanning tree)
with at least $k$ leaves in a given digraph. The problem has recently received much attention
from the viewpoint of parameterized algorithms~\cite{alonLNCS4596,AlonFGKS07fsttcs,BoDo2,KnLaRo}.
In this paper we step aside and take a kernelization based approach to the
{\sc $k$-Leaf-Out-Branching} problem. We give
the first polynomial kernel for
{\sc Rooted $k$-Leaf-Out-Branching}, a variant of {\sc $k$-Leaf-Out-Branching}
where the root of the tree searched for is also a part of the input. Our
kernel has cubic size and is obtained using extremal combinatorics.

For the {\sc $k$-Leaf-Out-Branching} problem we show that no polynomial
kernel is possible unless polynomial hierarchy collapses to third level 
by applying a recent breakthrough
result by Bodlaender et al.~\cite{BDFH08} in a non-trivial fashion. However our
positive results for {\sc Rooted $k$-Leaf-Out-Branching}
immediately imply that the seemingly
intractable the {\sc $k$-Leaf-Out-Branching} problem admits
a data reduction to $n$ independent $O(k^3)$ kernels. These two results,
tractability and intractability side by side, are the first separating
{\it many-to-one kernelization} from {\it Turing kernelization}. This answers
affirmatively an open problem  regarding ``cheat kernelization'' raised in~\cite{IWPECOPEN08}.

\end{abstract}

\section{Introduction}
Kernelization is a powerful and natural technique  in  the design of 
parameterized algorithms. 
The main idea of kernelization is to replace a given parameterized
instance $(I , k)$ of a problem $\Pi$ by a simpler instance $(I',k')$ of $\Pi$
in polynomial time,
such that $(I , k)$
is a yes instance if and only if $(I',k')$ is a yes instance and the size of $I'$ is bounded by a function of $k$ alone. The
reduced instance $I'$  is called the {\em kernel} for the problem.
Typically kernelization algorithms work by applying reduction rules, which iteratively reduce the instance to an
equivalent ``smaller" instance. From this point of view, kernelization can be seen as
pre-processing with an explicit performance guarantee, ``a
humble strategy for coping with hard problems, almost
universally employed" \cite{Fellows06}.

A parameterized problem is said to have a polynomial kernel if we have a
kernelization algorithm such that the size of the reduced instance obtained
as its output is bounded by a polynomial of the parameter of the input.
There are many parameterized problems for which polynomial, and even linear kernels
are known \cite{CKJ01,CFKX07,estivill,GN07,Thomasse09}. Notable examples include a $2k$-sized kernel for
{\sc $k$-Vertex Cover}~\cite{CKJ01}, a $O(k^2)$ kernel for {\sc $k$-Feedback Vertex Set}~\cite{Thomasse09} and a $67k$ kernel for
{\sc $k$-Planar-Dominating Set}~\cite{CFKX07}, among many others. While positive kernelization results have been around for
quite a while, the first results ruling out polynomial kernels for parameterized problems have appeared only recently.
In a seminal paper Bodlaender et al. \cite{BDFH08}
have shown that a variety of important FPT problems cannot have
polynomial kernels unless the polynomial hierarchy collapses to third level ($PH=\Sigma_p^3$), a well known
complexity theory hypothesis.  Examples of such  problems are
\textsc{$k$-Path}, $k$-\textsc{Minor Order Test}, $k$-\textsc{Planar Graph Subgraph Test}, and many others.
 However, while this negative result rules out the existence of a polynomial kernel for these problems it does not rule out
 the possibility of a kernelization algorithm reducing the instance to $|I|^{O(1)}$ independent polynomial kernels. 
This raises the question of the relationship between {\it many-to-one kernelization} and {\it Turing kernelization}, a question raised in~\cite{IWPECOPEN08,estivill,GN07}. That is,
 can we have a natural parameterized problem for which there is no polynomial kernel but
 we can ``cheat" this lower bound by providing $|I|^{O(1)}$ independent polynomial kernels.
Besides of theoretical interest, this type of results
would be  very desirable
from a practical point of view as well.
In this paper, we address the issue of many-to-one kernelization versus Turing kernelization
through  {\sc $k$-Leaf Out-Branching}.



The {\sc Maximum Leaf Spanning Tree} problem on connected undirected graphs is to find a spanning tree with the maximum number of leaves in a given input graph $G$.
The problem is well studied both from an algorithmic~\cite{GalbiatMM97,LuR98,Solis-Oba98,fominGK06} and combinatorial~\cite{DingJS01,GriggsW92,KleitmanW91,LinialS87} point of view. The problem has been studied
from the parameterized complexity perspective as well~\cite{bonsmaLNCS2747,estivill,FellowsMRS00}. 
An extension of {\sc
Maximum Leaf Spanning Tree} to directed graphs is defined as follows.
We say that a subdigraph $T$ of a digraph $D$ is an {\em out-tree}
if $T$ is an
oriented tree with only one vertex $r$ of in-degree zero (called
the {\em root}). The vertices of $T$ of out-degree zero are called {\em
leaves}. If $T$ is a spanning out-tree, i.e. $V(T)=V(D)$, then
$T$ is called an {\em out-branching} of $D$.
 The {\sc
Directed Maximum Leaf Out-Branching}  problem  is to find an
out-branching in a given digraph with the maximum number of leaves.
The parameterized version of the {\sc
Directed Maximum Leaf Out-Branching}  problem  is
{\sc
$k$-Leaf Out-Branching}, where one for a given digraph $D$ and integer
$k$ is asked to decide whether $D$ has an
out-branching  with at least $k$ leaves. If we replace out-branching with out-tree in
the definition of  {\sc $k$-Leaf Out-Branching} we get the problem of {\sc $k$-Leaf Out-Tree}.

Unlike its undirected counterpart, the study of
 {\sc $k$-Leaf Out-Branching} has only begun recently.
Alon et al.  \cite{alonLNCS4596,AlonFGKS07fsttcs} proved that the problem is
 fixed parameter tractable (FPT) by providing an algorithm
 deciding in time $O(f(k)n)$ whether a strongly connected digraph has an out-branching
with at least $k$ leaves. Bonsma and Dorn \cite{BoDo2} extended this result to connected
digraphs, and improved  the running time of the algorithm.
Very recently, Kneis et al.~\cite{KnLaRo} provided parameterized algorithm solving
the problem in time $4^kn^{O(1)}$. In a related work Drescher and Vetta~\cite{DV08} described an
$\sqrt{OPT}$-approximation algorithm for  the {\sc
Directed Maximum Leaf Out-Branching}  problem.
Let us remark, that despite of similarities of directed and
undirected variants
of  {\sc
Maximum Leaf Spanning Tree}, the directed case requires a totally different
approach.
The existence of a polynomial kernel for
 {\sc $k$-Leaf Out-Branching} has not been addressed until now.

\medskip\noindent\textbf{Our contribution.}
We prove that {\sc Rooted $k$-Leaf Out-Branching},
where for a given vertex $r$ one asks for $k$-leaf out-branching rooted at $r$,
admits a polynomial, in fact a $O(k^3)$, kernel. A similar result also holds for
{\sc Rooted $k$-Leaf Out-Tree}, where we are looking for a rooted (not necessary
spanning) tree with $k$ leaves. While many polynomial kernels are known
for undirected graphs, this is the first, to our knowledge, non-trivial parameterized
problem on digraphs admitting a polynomial kernel. To obtain the kernel we
establish a number of results on the structure of digraphs not having a $k$-leaf
out-branching. These results may be of independent interest.

In light of our positive results it is natural to suggest that {\sc $k$-Leaf Out-Branching} admits polynomial kernel as well. We
find it a bit striking that this is not the case. We establish
kernelization lower bounds 
by proving that unless   $PH=\Sigma_p^3$,
there is no polynomial kernel for neither {\sc $k$-Leaf Out-Branching} nor
{\sc $k$-Leaf Out-Tree}. While the main idea of our proof is based on the framework
of Bodlaender et al. \cite{BDFH08}, our adaptation is non-trivial. 
In particular, we use the cubic kernel obtained for {\sc Rooted $k$-Leaf Out-Branching}
to prove the lower bound. Our contributions are summarized in Table~\ref{table:results}.

\begin{table}[t]\label{table:results}
\begin{center}
\begin{tabular}{|c|c|c|}
\hline
 & \textsc{$k$-Out-Tree} & \textsc{$k$-Out-Branching} \\
\hline \hline
Rooted & $O(k^3)$ kernel & $O(k^3)$ kernel\\
\hline
\multirow{2}{*}{Unrooted} & No $poly(k)$ kernel, & No $poly(k)$ kernel, \\
 & $n$~kernels of size $O(k^3)$ & $n$~kernels of size $O(k^3)$\\
\hline
\end{tabular}
\end{center}
\caption{\label{table:pd} Our Results}
\end{table}

Finally, let us remark that the polynomial kernels
for the rooted versions of our problems 
provide a ``cheat" solution
for the poly-kernel-intractable {\sc $k$-Leaf Out-Branching} and
{\sc $k$-Leaf Out-Tree}. Indeed, let $D$ be a digraph on $n$ vertices.
By running the kernelization for the rooted
version of the problem for every vertex of $D$ as a root,
we obtain $n$ graphs where each of them has $O(k^3)$ vertices, such that at least
one of them has a $k$-leaf out-branching if and only if $D$ does.


\section{Preliminaries}

Let $D$ be a directed graph or digraph for short.
By $V(D)$ and $A(D)$ we represent
vertex set and arc set respectively of $D$.
Given a subset $V'\subseteq V(D)$ of a digraph $D$, by
$D[V']$ we mean the digraph induced on $V'$.  A vertex $y$ of $D$ is an
{\em in-neighbor} ({\em out-neighbor}) of a vertex $x$ if $yx\in A$
($xy\in A$). The {\em in-degree} ({\em out-degree}) of a vertex $x$ is the number of its in-neighbors
(out-neighbors) in $D$. Let  $P=p_1p_2\ldots p_l$ be a given path. Then by  $P[p_ip_j]$ we denote
a subpath of $P$ starting at  vertex $p_i$ and ending at vertex $p_j$. For a given vertex $q\in V(D)$,
by $q$-out-branching (or $q$-out-tree)  we denote an out-branching (out-tree) of $D$ rooted at vertex
$q$.

We say that the removal of an arc $uv$ (or a vertex set $S$) \emph{disconnects} a vertex $w$
from the root $r$ if every path from $r$ to $w$ in $D$ contains arc $uv$ (or one of the vertices in $S$).
An arc $uv$ is contracted as follows, add a new vertex $u'$, and for each arc $wv$ or $wu$ add
the arc $wu'$ and for an arc $vw$ or $uw$ add the arc $u'w$,
remove all arcs incident to $u$ and $v$ and the vertices $u$ and $v$.
We say that a reduction rule is \emph{safe} for a value $k$ if whenever the rule is applied to an instance $(D,k)$ to obtain an instance $(D',k')$, $D$ has an $r$-out-branching with at least $k$ leaves if and only if $D'$ has
an $r$-out-branching with at least $k'$ leaves. 
We also need the following.
\begin{proposition}\label{prop:grow}{\rm \cite{KnLaRo}}
Let $D$ be a digraph and $r$ be a vertex from which every vertex in $V(D)$ is reachable. Then if we have an
out-tree rooted at $r$ with $k$ leaves  then we also have an out-branching rooted at $r$ with $k$ leaves.
\end{proposition}
Let $T$ be an out-tree of a digraph $D$. We say that $u$ is a {\em parent} of 
$v$ and $v$ is a {\em child} of $u$ if $uv\in A(T)$. We say that $u$ is an 
{\em ancestor} of $v$ if there is a directed path from $u$ to $v$ in $T$. An 
arc $uv$ in $A(D)\setminus A(T)$ is called a {\em forward} arc if  $u$ is an 
ancestor of $v$, a {\em backward} arc if $v$ is an ancestor of $u$ and a 
{\em cross} arc otherwise.  Finally, parameterized decision problems are defined by specifying the input ($I$), the parameter ($k$), and
the question to be answered.  A parameterized problem that can be solved in time $f(k)|I|^{O(1)}$ where $f$ is a
function of $k$ alone is said to be fixed parameter tractable (FPT).

\section{Reduction Rules for {\sc Rooted $k$-Leaf Out-Branching}}
\label{section:redrule}
In this section we give all the data reduction rules we apply on the given instance of
{\sc Rooted $k$-Leaf Out-Branching} to shrink its size.


\begin{RedRule} \label{reachRule}{\rm[Reachability Rule]}
If there exists a vertex $u$ which is disconnected from the root $r$, then return {\sc No}.
\end{RedRule}
\noindent
For the {\sc Rooted $k$-Leaf Out-Tree} problem the Rule~\ref{reachRule} translates into following:
If a vertex $u$ is disconnected from the root $r$, then remove $u$ and all in-arcs to $u$ and out-arcs from $u$.

\begin{RedRule} \label{uslessArcRule}{\rm [Useless arc Rule]}
If vertex $u$ disconnects a vertex $v$ from the root $r$,
then remove the arc $vu$.
\end{RedRule}

\begin{lemma}$[\star]$\footnote{Proofs of the results labeled with $[\star]$ 
have been moved to appendix due to space restrictions.} 
Reduction Rules \ref{reachRule} and \ref{uslessArcRule} are safe.
\end{lemma}

\begin{RedRule}
\label{bridgeRule}{\rm [Bridge Rule]}
If an arc $uv$ disconnects at least two vertices from the root $r$, contract arc $uv$.
\end{RedRule}

\begin{lemma}$[\star]$
 Reduction Rule \ref{bridgeRule} is safe.
\end{lemma}

\begin{RedRule} \label{AvoidableArcRule}{\rm [Avoidable Arc Rule]}
If a vertex set $S$, $|S| \leq 2$, disconnects a vertex $v$ from the root $r$, $vw \in A(D)$ and $xw \in A(D)$ for all $x \in S$, then delete the arc $vw$.
\end{RedRule}

\begin{lemma}$[\star]$ 
Reduction Rule \ref{AvoidableArcRule} is safe.
\end{lemma}

\begin{figure}
\begin{center}
\input{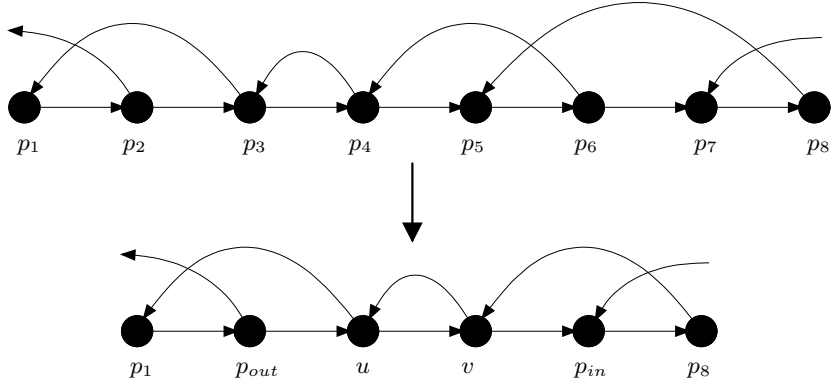}
\end{center}
\caption{\label{fig:redrule5}An Illustration of Reduction Rule ~\ref{twoDirectionalPathRule}.}

\end{figure}

\begin{RedRule} \label{twoDirectionalPathRule}{\rm [Two directional path Rule]}
If there is a path $P = p_1p_2 \ldots p_{l-1}p_l$ with $l = 7$ or $l = 8$ such that
\begin{itemize}
\setlength{\itemsep}{-2pt}
 \item $p_1$ and $p_{in} \in \{p_{l-1},p_l\}$ are the only vertices with in-arcs from the outside of $P$.
 \item $p_l$ and $p_{out} \in \{p_1,p_2\}$ are the only vertices with out-arcs to the outside of $P$.
 \item The path $P$ is the unique out-branching of $D[V(P)]$ rooted at $p_1$.
 \item There is a path $Q$ that is the unique out-branching of $D[V(P)]$ rooted at $p_{in}$.
 \item The vertex after $p_{out}$ on $P$ is not the same as the vertex after $p_l$ on $Q$.
\end{itemize}
Then delete $R = P \setminus \{p_1, p_{in}, p_{out}, p_l\}$ and all arcs incident to these vertices from $D$.
Add two vertices $u$ and $v$ and the arc set $\{p_{out}u, uv, vp_{in}, p_lv, vu, up_{1}\}$ to $D$.
\end{RedRule}

Notice that every vertex on $P$ has in-degree at most $2$ and out-degree at most $2$. Figure~\ref{fig:redrule5} gives an example of an application of 
Reduction Rule~\ref{twoDirectionalPathRule}. 

\begin{lemma}
Reduction Rule \ref{twoDirectionalPathRule} is safe. \end{lemma}
\begin{proof}
Let $D'$ be the graph obtained by performing Reduction Rule \ref{twoDirectionalPathRule} to a path $P$ in $D$. Let $P_u$ be the path $p_1p_{out}uvp_{in}p_l$ and $Q_v$ be the path $p_{in}p_lvup_1p_{out}$. Notice that $P_u$ is the unique out-branching of $D'[V(P_u)]$ rooted at $p_1$ and that $Q_v$ is the unique out-branching of $D'[V(P_u)]$ rooted at $p_{in}$.

Let $T$ be an $r$-out-branching of $D$ with at least $k$ leaves. Notice that since $P$ is the unique out-branching of $D[V(P)]$ rooted at $p_1$, $Q$ is the unique out-branching of $D[V(P)]$ rooted at $p_{in}$ and $p_1$ and $p_{in}$ are the only vertices with in-arcs from the outside of $P$, $T[V(P)]$ is either a path or the union of two vertex disjoint paths. Thus, $T$ has at most two leaves in $V(P)$ and at least one of the following three cases must apply.

\begin{enumerate}
\setlength{\itemsep}{-2pt}
 \item $T[V(P)]$ is the path $P$ from $p_1$ to $p_l$.
 \item $T[V(P)]$ is the path $Q$ from $p_{in}$ to $p_{out}$.
 \item $T[V(P)]$ is the vertex disjoint union of a path $\tilde{P}$ that is a subpath of $P$ rooted at $p_1$, and a path $\tilde{Q}$ that is a subpath of $Q$ rooted at $p_{in}$.
\end{enumerate}

In the first case we can replace the path $P$ in $T$ by the path $P_u$ to get an $r$-out-branching of $D'$ with at least $k$ leaves. Similarly, in the second case, we can replace the path $Q$ in $T$ by the path $Q_v$ to get an $r$-out-branching of $D'$ with at least $k$ leaves. For the third case, observe that $\tilde{P}$ must contain $p_{out}$ since $p_{out} = p_1$ or $p_1$ appears before $p_{out}$ on $Q$ and thus, $p_{out}$ can only be reached from $p_1$. Similarly, $\tilde{Q}$ must contain $p_l$. Thus, $T \setminus R$ is an $r$-out-branching of $D \setminus R$. We build an $r$-out-branching $T'$ of $D'$ by taking $T \setminus R$ and letting $u$ be the child of $p_{out}$ and $v$ be the child of $p_l$. In this case $T$ and $T'$ have same number of leaves outside of $V(P)$ and $T$ has at most two leaves in $V(P)$ while both $u$ and $v$ are leaves in $T'$. Hence $T'$ has at least $k$ leaves.

The proof for the reverse direction is similar and can be found 
in Appendix~\ref{appendix:redrule}.
\end{proof}

We say that a digraph $D$ is a \emph{reduced instance} of
\textsc{Rooted $k$-Leaf Out-Branching} if none of the
reduction rules (Rules $1$--$5$) can be applied to $D$.
It is easy to observe from the description of the reduction rules  that we can apply them in polynomial time, resulting in the
following lemma.
\begin{lemma}
\label{lemma:polyreduction}
For a digraph $D$ on $n$ vertices we can obtain a reduced instance $D'$ in polynomial time.
\end{lemma}

\section{Polynomial Kernel: Bounding a Reduced No-instance}
\label{section:polykernelno}
In this section we show that any reduced no-instance of
\textsc{Rooted $k$-Leaf Out-Branching} must have at most $O(k^3)$ vertices. In order to do so
we start with $T$, a BFS-tree rooted at $r$, of a reduced instance $D$ and look at a path $P$
of $T$ such that every vertex on $P$ has out-degree one in $T$. 

We bound the number of endpoints of arcs with one endpoint in $P$ and one endpoint outside of $P$ (Section \ref{subsec:inandoutnbh}). We then
use these results to bound the size of any maximal path with every vertex having out-degree one in $T$
(Section~\ref{subsection:boundpaths}). Finally, we combine these results to bound the size of any reduced no-instance of
\textsc{Rooted $k$-Leaf Out-Branching} by $O(k^3)$.

\subsection{Bounding the Number of Entry and Exit Points of a Path}
\label{subsec:inandoutnbh}

Let $D$ be a reduced no-instance, and $T$ be a BFS-tree rooted at $r$.
The BFS tree $T$ has at most $k-1$ leaves and hence at most $k-2$ vertices
with out-degree at least $2$ in $T$. Now, let $P = p_1p_2 \ldots p_l$
be a path in $T$ such that all vertices in $V(P)$ have out-degree $1$ in $T$ ($P$ does not need to be a maximal path of $T$).
Let $T_1$ be the subtree of $T$ induced by the vertices reachable from $r$ in $T$
without using vertices in $P$ and let $T_2$ be the subtree of $T$
rooted at the child $r_2$ of $p_l$ in $T$.
Since $T$ is a BFS-tree, it does not have any forward arcs, and thus $p_lr_2$ is the only arc from $P$ to $T_2$.
Thus all arcs originating in $P$ and ending outside of $P$ must have their endpoint in $T_1$.

\begin{lemma} \label{le:missPath}
Let $D$ be a reduced instance,
$T$ be a BFS-tree rooted at $r$, and
$P = p_1p_2 \ldots p_l$ be a path in $T$ such that all vertices in $V(P)$ have out-degree $1$ in $T$.
Let $up_{i} \in A(D)$, for some $i$ between $1$ and $l$, be an arc with $u \notin P$.
There is a path $P_{up_{i}}$ from $r$ to $p_i$ using the arc $up_i$,
such that $V(P_{up_i}) \cap V(P) \subseteq \{p_i, p_l\}$.
\end{lemma}
\begin{proof}
Let $T_1$ be the subtree of $T$ induced by the vertices reachable from $r$ in $T$
without using vertices in $P$ and let $T_2$ be the subtree of $T$ rooted at the child $r_2$ of $p_l$ in $T$.
If $u \in V(T_1)$ there is a path from $r$ to $u$ avoiding $P$.
Appending the arc $up_i$ to this path yields the desired path $P_{up_i}$,
so assume $u \in V(T_2)$. If all paths from $r$ to $u$ use the arc $p_{l-1}p_l$
then $p_{l-1}p_l$ is an arc disconnecting $p_l$ and $r_2$ from $r$,
contradicting that Reduction Rule \ref{bridgeRule} can not be applied.
Let $P'$ be a path from $r$ to $u$ not using the arc $p_{l-1}p_l$.
Let $x$ be the last vertex from $T_1$ visited by $P'$.
Since $P'$ avoids $p_{l-1}p_l$ we know that $P'$ does not visit any vertices of $P \setminus \{p_l\}$ after $x$.
We obtain the desired path $P_{up_i}$ by taking the path from $r$ to $x$ in $T_1$
followed by the subpath of $P'$ from $x$ to $u$ appended by the arc $up_i$.
\end{proof}

\begin{corollary} \label{cor:fewInNodes}
Let $D$ be a reduced no-instance,
$T$ be a BFS-tree rooted at $r$ and $P = p_1p_2 \ldots p_l$
be a path in $T$ such that all vertices in $V(P)$ have out-degree $1$ in $T$.
There are at most $k$ vertices in $P$ that are endpoints of arcs originating outside of $P$.
\end{corollary}

\begin{proof}
Let $S$ be the set of vertices in $P \setminus \{p_l\}$
that are endpoints of arcs originating outside of $P$.
For the sake of contradiction suppose that there are at least $k+1$ vertices in $P$
that are endpoints of arcs originating outside of $P$. Then $|S| \geq k$.
By Lemma~\ref{le:missPath} there exists a path from the root $r$ to every vertex in $S$,
that avoids vertices of $P \setminus \{p_l\}$ as an intermediate vertex.
Using these paths we can build an $r$-out-tree with every vertex in $S$ as a leaf.
This $r$-out-tree can be extended to a $r$-out-branching with at least $k$
leaves by Proposition~\ref{prop:grow}, contradicting that $D$ is a no-instance.
\end{proof}

\begin{lemma} \label{le:fewOutVertices}
Let $D$ be a reduced no-instance, $T$ be a BFS-tree rooted at $r$ and
$P = p_1p_2 \ldots p_l$ be a path in $T$ such that all
vertices in $V(P)$ have out-degree $1$ in $T$.
There are at most $7(k-1)$ vertices outside of $P$ that are endpoints of arcs originating in $P$.
\end{lemma}
\begin{proof}
Let $X$ be the set of vertices outside $P$ which are out-neighbors of the vertices on $P$. Let $P'$ be the path from $r$ to $p_1$ in $T$ and $r_2$ be the unique child
of $p_l$ in $T$. First, observe that since there are no forward arcs, $r_2$ is the only out-neighbor of vertices in $V(P)$ in the subtree of $T$
rooted at $r_2$. In order to bound the size of $X$, we differentiate between two kinds of out-neighbors of vertices on $P$.
\begin{itemize}
\setlength{\itemsep}{-2pt}
\item Out-neighbors of $P$ that are not in $V(P')$.
\item Out-neighbors of $P$ in $V(P')$.
\end{itemize}
First, observe that $|X \setminus V(P')| \leq k-1$. Otherwise we could have made an $r$-out-tree with at least $k$ leaves by taking the path $P'P$ and
adding $X \setminus V(P')$ as leaves with parents in $V(P)$.

In the rest of the proof we bound $|X \cap V(P')|$. 
Let $Y$ be the set of vertices on $P'$ with out-degree at least $2$ in $T$ and let
$P_1,P_2,\ldots,P_t$ be the remaining subpaths of $P'$ when vertices in $Y$ are removed.
For every $i \leq t$, $P_i=v_{i1}v_{i2}\ldots v_{iq}$. We define the vertex set $Z$ to contain the two last
vertices of each path $P_i$. The number of vertices with out-degree at least $2$ in $T$ is upper bounded by
$k-2$ as $T$ has at most $k-1$ leaves. Hence, $|Y|\leq k-2$, $t \leq k-1$ and $|Z|\leq 2(k-1)$.

%

\begin{claim}
\label{claim:arcexists}
For every path $P_i=v_{i1}v_{i2}\ldots v_{iq}$, $1\leq i\leq t$, there is either an arc $u_iv_{iq-1}$ or  $u_iv_{iq}$ where $u_i\notin V(P_i)$.
\end{claim}
To see the claim observe that the removal of arc $v_{iq-2}v_{iq-1}$ does not disconnect the root $r$
from both $v_{iq-1}$ and $v_{iq}$ else Rule \ref{bridgeRule} would have been applicable to our reduced
instance. For brevity assume that $v_{iq}$ is reachable from $r$ after the removal of arc $v_{iq-2}v_{iq-1}$. Hence
there exists a path from $r$ to $v_{iq}$. Let $u_iv_{iq}$ be the last arc of this path. The fact that the BFS tree $T$ does not
have any forward arcs implies that  $u_i\notin V(P_i)$.

To every path  $P_i=v_{i1}v_{i2}\ldots v_{iq}$, $1\leq i\leq t$, we associate an interval $I_i=v_{i1}v_{i2}\ldots v_{iq-2}$ 
and an arc $u_iv_{iq'}$, $q'\in \{q-1,q\}$. This arc exists by Claim~\ref{claim:arcexists}.
Claim~\ref{claim:arcexists} and Lemma~\ref{le:missPath} together imply that for every path $P_i$ there is a path $P_{ri}$
from the root $r$ to $v_{iq'}$ that does not use any vertex in $V(P_i) \setminus \{v_{iq-1},v_{iq}\}$ as an intermediate vertex.
That is, $V(P_{ri}\cap (V(P_i) \setminus \{v_{iq-1},v_{iq}\})=\emptyset$.

Let $P_{ri}'$ be a subpath of $P_{ri}$ starting at a vertex $x_i$ before $v_{i1}$ on $P'$ and ending in a vertex $y_i$ after $v_{iq-2}$ on $P'$.
We say that a path $P_{ri}'$ {\em covers} a vertex $x$ if $x$ is on the subpath of $P'$
between $x_i$ and $y_i$ and we say that it {\em covers} an interval $I_j$ if $x_i$ appears before $v_{j1}$ on the path $P'$
and $y_i$ appears after $v_{jq-2}$ on $P'$. Observe that the path $P_{ri}'$ covers the interval $I_i$.

Let ${\cal P}= \{P_1',P_2',\ldots,P_l'\}\subseteq \{P_{r1}',\ldots, P_{rt}'\}$ be a minimum collection of paths,
such that every interval $I_i$, $1\leq i \leq t$, is covered by at least one of the paths in $\cal P$.
Furthermore, let the paths of ${\cal P}$ be numbered by the appearance of their first vertex on $P'$.
The minimality of $\cal P$ implies that for every $P_i' \in \cal P$ there is an interval $I_i'\in \{I_1,\ldots,I_t\}$
such that $P_i'$ is the only path in $\cal P$ that covers $I_i'$.
\begin{claim}
For every $1\leq i \leq l$, no vertex of $P'$ is covered by both  $P_i'$ and $P_{i+3}'$.
\end{claim}
The path $P_{i+1}'$ is the only path in $\cal P$
that covers the interval $I_{i+1}'$ and hence $P_i'$ does not cover the last vertex of $I_{i+1}'$.
Similarly $P_{i+2}'$ is the only path in $\cal P$ that covers the interval $I_{i+2}'$ and hence
$P_{i+3}'$ does not cover the first vertex of $I_{i+2}'$. Thus the set of vertices covered by both $P_i'$ and $P_{i+3}'$ is empty.

Since paths $P_i'$ and $P_{i+3}'$ do not cover a common vertex, we have that
the end vertex of $P_i'$ appears before the start vertex of $P_{i+3}'$
on $P'$ or is the same as the start vertex of $P_{i+3}'$. Partition the paths
of ${\cal P}$ into three sets ${\cal P}_0,{\cal P}_1,{\cal P}_2$, where path $P_i' \in {\cal P}_{i\,mod\,3}$.
Also let ${\cal I}_i$ be the set of intervals covered by ${\cal P}_i$. Observe that every interval $I_j$,
$1\leq j \leq t$, is part of some ${\cal I}_i$ for $i\in \{0,1,2\}$.

Let $i \leq 3$ and consider an interval $I_j \in {\cal I}_i$. There is a path $P_{j'} \in {\cal P}_i$ that
covers $I_j$ such that both endpoints of $P_{j'}$ and none of the inner vertices of $P_{j'}$ lie on $P'$.
Furthermore for any pair of paths $P_a$, $P_b \in {\cal P}_i$ such that $a < b$, there is a subpath in $P'$ from the
endpoint of $P_a$ to the starting point of $P_b$. Thus for every $i \leq 3$ there is a path $P_{i}^{*}$ from the root $r$ to $p_1$
which does not use any vertex of the intervals covered by the paths in ${\cal P}_i$.


We now claim that the total number of vertices on intervals $I_j$, $1\leq j \leq t$,
which are out-neighbors of vertices on $V(P)$ is bounded by $3(k-1)$. If not, then for some
$i$, the number of out-neighbors in ${\cal I}_i$ is at least $k$. Now we can make an $r$-out-tree
with $k$ leaves by taking any $r$-out-tree in $D[V(P_i^*)\cup V(P)]$ and adding the out-neighbors of the
vertices on $V(P)$ in ${\cal I}_i$ as leaves with parents in $V(P)$. 

Summing up the obtained upper bounds yields $|X|\leq (k-1)+|\{r_2\}|+|Y|+|Z|+3(k-1)\leq (k-1)+1+(k-2)+2(k-1)+3(k-1)= 7(k-1)$, concluding the proof.
%
\end{proof}

\noindent\textbf{Remark:} Observe that the path $P$ used in Lemmas~\ref{le:missPath} and ~\ref{le:fewOutVertices} and Corollary~\ref{cor:fewInNodes}
need not be a maximal path in $T$ with its vertices having out-degree one in $T$.

\subsection{Bounding the Length of a Path: On Paths through Nice Forests}
\label{subsection:boundpaths}

For a reduced instance $D$, a BFS tree $T$ of $D$ rooted at $r$, let $P = p_1p_2 \ldots p_l$ be a path in $T$ such that all vertices in $V(P)$ have out-degree  $1$ in $T$, and let $S$ be the set of vertices in $V(P) \setminus \{p_l\}$ with an in-arc from the outside of $P$.

\begin{definition}
A subforest $F=(V(P),A(F))$ of $D[V(P)]$ is said to be \emph{nice forest of $P$} if the following three properties are satisfied: 
(a) $F$ is a forest of directed trees rooted at vertices in $S$; 
(b)  If $p_ip_j \in A(F)$ and $i < j$ then $p_i$ has out-degree at least $2$ in $F$ or $p_j$ has in-degree $1$ in $D$; and (c) If $p_ip_j \in A(F)$ and $i > j$ then for all $q>i$, $p_qp_j \notin A(D)$.

\end{definition}

In order to bound the size a reduced no-instance $D$ we are going to consider a nice forest with the maximum number of leaves. However, in order to do this, we first need to show the existence of a nice forest of $P$.

In the following discussion let $D$ be a reduced no-instance, $T$ be a BFS tree $T$ of $D$ rooted at $r$,  $P = p_1p_2 \ldots p_l$ be a path in $T$ such that all vertices in $V(P)$ have out-degree  $1$ in $T$ and $S$ be the set of vertices in $V(P) \setminus \{p_l\}$ with an in-arc from the outside of $P$.

\begin{lemma} $[\star]$
\label{lemma:niceforest}
There is a nice forest in $P$.
\end{lemma}

For a nice forest $F$ of $P$, we define the set of \emph{key} vertices of $F$ to be the set of vertices in $S$, the leaves of $F$, the vertices of $F$ with out-degree at least $2$ and the set of vertices whose parent in $F$ has out-degree at least $2$.

\begin{lemma} $[\star]$ 
\label{le:keyVertexCount} Let $F$ be a nice forest of $P$. There are at most $5(k-1)$ key vertices of $F$. \end{lemma}

We can now turn our attention to a nice forest $F$ of $P$ with the maximum number of leaves. Our goal is to show that if the key points of $F$ are to spaced out on $P$ then some of our reduction rules must apply. First, however, we need some more observations about the interplay between $P$ and $F$.

\begin{obs} \label{obs:uniquePath}{\rm [Unique Path]}
For any two vertices $p_i$, $p_j$ in $V(P)$ such that $i < j$,
$p_ip_{i+1}\ldots p_j$ is the only path from $p_i$ to $p_j$ in $D[V(P)]$.
\end{obs}
\begin{proof}
As $T$ is a BFS-tree it has no forward arcs.
So any vertex set $X = \{p_1, p_2, \ldots, p_q\}$ with $q < |V(P)|$,
the arc $p_qp_{q+1}$ is the only arc in $D$ from a vertex in $X$ to a vertex in $V(P) \setminus X$.
\end{proof}

\begin{corollary} $[\star]$ \label{cor:noForward}
No arc $p_{i}p_{i+1}$ is a forward arc of $F$.
\end{corollary}

\begin{obs} \label{obs:longJump}
Let $p_tp_j$ be an arc in $A(F)$ such that neither $p_t$ nor $p_j$ are key vertices,
and $t \in \{j-1,j+1,\ldots,l\}$.
Then for all $q>t$, $p_qp_j \not\in A(D)$.
\end{obs}
Observation~\ref{obs:longJump} follows directly from the definitions of a 
nice forest and key vertices.

\begin{obs} $[\star]$\label{obs:noTwoFront}
If neither $p_i$ nor $p_{i+1}$ are key vertices,
then either $p_ip_{i+1} \notin A(F)$ or $p_{i+1}p_{i+2} \notin A(F)$.
\end{obs}

In the following discussion let $F$ be a nice forest of $P$ with the maximum number of leaves and
let $P'= p_xp_{x+1}\ldots p_y$ be a subpath of $P$ containing no key vertices,
and additionally having the property that $p_{x-1}p_x \notin A(F)$ and $p_{y}p_{y+1} \notin A(F)$.

\begin{lemma} $[\star]$ \label{le:FInducesPath}
$V(P')$ induces a directed path in $F$.
\end{lemma}

In the following discussion let $Q'$ be the directed path $F[V(P')]$.

\begin{obs} $[\star]$\label{obs:PFOrder}
For any pair of vertices $p_i,p_j \in V(P')$ if $i \leq j-2$
then $p_j$ appears before $p_i$ in $Q'$.
\end{obs}

\begin{lemma} $[\star]$\label{le:PAndFIsAll}
All arcs of $D[V(P')]$ are contained in $A(P') \cup A(F)$.
\end{lemma}

\begin{lemma} \label{le:twoInVertices}
If $|P'| \geq 3$ there are exactly $2$ vertices in $P'$
that are endpoints of arcs starting outside of $P'$.
\end{lemma}

\begin{proof}
By Observation \ref{obs:uniquePath}, $p_{x-1}p_x$ is the only arc between $\{p_1,p_2,\ldots,p_{x-1}\}$ and $P'$.
By Lemma \ref{le:FInducesPath}, $F[V(P')]$ is a directed path $Q'$.
Let $p_q$ be the first vertex on $Q'$ and notice that the parent of $p_q$ in $F$ is outside of $V(P')$.
Observation \ref{obs:PFOrder} implies that $q \geq y-1$.
Hence $p_q$ and $p_x$ are two distinct vertices that are endpoints of arcs starting outside of $P'$.
It remains to prove that they are the only such vertices.
Let $p_i$ be any vertex in $P' \setminus \{p_x, p_q\}$.
By Lemma \ref{le:FInducesPath} $V(P')$ induces a directed path $Q'$ in $F$,
and since $p_q$ is the first vertex of $Q'$, the parent of $p_i$ in $F$ is in $V(P')$.
Observation \ref{obs:longJump} yields then that $p_tp_i \not\in A(D)$ for any $t > y$.
\end{proof}

\begin{obs}$[\star]$ \label{obs:uniqueBackPath}
Let $Q' = F[V(P')]$.
For any pair of vertices $u,v$ such that there is a path $Q'[uv]$ from $u$ to $v$ in $Q'$,
$Q'[uv]$ is the unique path from $u$ to $v$ in $D[V(P')]$.
\end{obs}

\begin{lemma} \label{le:twoArcsOut}
For any vertex $x \notin V(P')$ there are at most $2$ vertices in $P'$ with arcs to $x$.
\end{lemma}
\begin{proof}
Suppose there are $3$ vertices $p_a, p_b, p_c$ in $V(P')$
such that $a < b < c$ and such that $p_ax, p_bx, p_cx \in A(D)$.
By Lemma \ref{le:FInducesPath} $Q' = F[V(P')]$ is a directed path.
If $p_a$ appears before $p_b$ in $Q'$ then Observation \ref{obs:PFOrder} implies
that $a + 1 = b$ and that $p_b$ has in-degree $1$ in $D$.
Then $p_a$ separates $p_b$ from the root and hence
Rule \ref{AvoidableArcRule} can be applied to remove the arc $p_bx$ contradicting that $D$ is a reduced instance.
Hence $p_b$ appears before $p_a$ in $Q'$.
By an identical argument $p_c$ appears before $p_b$ in $Q'$.

Let $P_b$ be a path in $D$ from the root to $p_b$ and
let $u$ be the last vertex in $P_b$ outside of $V(P')$.
Let $v$ be the vertex in $P_b$ after $u$.
By Lemma \ref{le:twoInVertices}, $u$ is either $p_x$ or the first vertex $p_q$ of $Q'$.
If $u = p_x$ then Observation \ref{obs:uniquePath} implies that $P_b$ contains $p_a$,
whereas if $u = p_q$ then Observation \ref{obs:uniqueBackPath} implies that $P_b$ contains $p_c$.
Thus the set $\{p_a, p_c\}$ separates $p_b$ from the root and
hence Rule \ref{AvoidableArcRule} can be applied to remove the arc $p_bx$ contradicting that
$D$ is a reduced instance.
\end{proof}

\begin{corollary} $[\star]$
\label{cor:fewArcsOut}
There are at most $14(k-1)$ vertices in $P'$ with out-neighbors 
outside of $P'$. 
\end{corollary}

\begin{lemma} \label{le:boundPPrime}
$|P'| \leq 154(k-1)+10$.
\end{lemma}
\begin{proof}
Assume for contradiction that $|P'| > 154(k-1)+10$ and 
let $X$ be the set of vertices in $P'$ with arcs to vertices outside of $P'$. 
By Corollary \ref{cor:fewArcsOut}, $|X| \leq 14(k-1)$. 
Hence there is a subpath of $P'$ on at least 
$154(k-1)+10 / (14(k-1)+1) = 9$ vertices containing no vertices of $X$.
By Observation \ref{obs:noTwoFront} there is a subpath $P''=p_ap_{a+1}\ldots p_b$ of $P'$ on 
$7$ or $8$ vertices such that neither $p_{a-1}p_a$ nor $p_{b}p_{b+1}$ are arcs of $F$. By Lemma \ref{le:FInducesPath} $F[V(P'')]$ is a directed path $Q''$.
Let $p_q$ and $p_t$ be the first and last vertices of $Q''$ respectively.
By Lemma \ref{le:twoInVertices} $p_a$ and $p_q$ are
the only vertices with in-arcs from outside of $P''$.
By Observation \ref{obs:PFOrder} $p_q \in \{p_{b-1}, p_b\}$ and $p_t \in \{p_a, p_{a+1}\}$.
By the choice of $P''$ no vertex of $P''$ has an arc to a vertex outside of $P'$.
Furthermore, since $P''$ is a subpath of $P'$ and $Q''$ is a subpath of $Q'$ Lemma \ref{le:PAndFIsAll} implies that $p_b$ and $p_t$ are the only vertices of $P'$ with out-arcs to the outside of $P''$.
By Lemma \ref{obs:uniquePath}, the path $P''$ is the unique out-branching of $D[V(P'')]$ rooted at $p_a$.
By Lemma \ref{obs:uniqueBackPath}, the path $Q''$ is the unique out-branching of $D[V(P'')]$ rooted at $p_q$.
By Observation \ref{obs:PFOrder} $p_{b-2}$ appears before $p_{a+2}$ in $Q''$ and hence the vertex after $p_b$ in $Q''$ and $p_{t+1}$ is not the same vertex.
Thus Rule \ref{twoDirectionalPathRule} can be applied on $P''$, contradicting that $D$ is a reduced instance.
\end{proof}

\begin{lemma} $[\star]$
\label{lemma:boundpath}
Let $D$ be a reduced no-instance to \textsc{Rooted $k$-Leaf Out-Branching}. 
Then $|V(D)| = O(k^3)$. 
\end{lemma}

\noindent
Lemma~\ref{lemma:boundpath} results in cubic kernel for {\sc Rooted $k$-Leaf Out-Branching} as follows.
 \begin{theorem}
 \label{thm:polykernel}
{\sc Rooted $k$-Leaf Out-Branching} and {\sc Rooted $k$-Leaf Out-Tree} 
admits a kernel of size $O(k^3)$.
\end{theorem}
\begin{proof}
Let $D$ be the reduced instance of {\sc Rooted $k$-Leaf Out-Branching}  obtained in polynomial time using 
Lemma~\ref{lemma:polyreduction}. If the size of $D$ is more than $1540k^3$ then return {\sc Yes}. Else we have an instance of size 
bounded by $O(k^3)$. The correctness of this step follows from Lemma~\ref{lemma:boundpath} which shows that any 
reduced no-instance to {\sc Rooted $k$-Leaf Out-Branching} has size bounded by $O(k^3)$. The result for 
{\sc Rooted $k$-Leaf Out-Tree} follows similarly. 
\end{proof}

\section{Kernelization Lower Bounds}
\label{cheatKernel}
In the last section we gave a cubic kernel for \textsc{Rooted $k$-Leaf Out-Branching}, it is natural to ask whether the closely related \textsc{$k$-Leaf Out-Branching} has a polynomial kernel. The answer to this question, somewhat surprisingly, is no, unless an unlikely collapse of complexity classes occurs. To show this we utilize a recent result of Bodlaender et al.~\cite{BDFH08} that states that any \emph{compositional} parameterized problem does not have a polynomial kernel unless the polynomial hierarchy collapses to the third level.

\begin{definition}[Composition \cite{BDFH08}]
A composition algorithm for a parameterized problem $L\subseteq \Sigma^* \times \mathbb N$ is an algorithm that
\begin{itemize}
\setlength{\itemsep}{-2pt}
\item receives as input a sequence $((x_1,k),\ldots,(x_t,k))$, with $(x_i,k)\in \Sigma^* \times \mathbb{N}^+$ for each $1\leq i \leq t$,
\item uses time polynomial in $\sum_{i=1}^t |x_i|+k$,
\item and outputs $(y,k')\in \Sigma^* \times \mathbb{N}^+$ with
\begin{enumerate}
\item  $(y,k')\in L ~\iff~ (x_i,k) \in L$ for some $1\leq i \leq t$.
\item $k'$ is polynomial in $k$.
\end{enumerate}
\end{itemize}
A parameterized problem is compositional if there is a composition algorithm for it.
\end{definition}
Now we state the main result of~\cite{BDFH08} which we need for our purpose.

\begin{theorem}[\cite{BDFH08}]
\label{thm:kernellbound}
Let $L$ be a compositional parameterized language whose unparameterized version
$\widetilde{L}$ is {\sc NP}-complete then unless {\rm PH=$\Sigma_p^3$}, there is no polynomial kernel for $L$.
\end{theorem}

\begin{theorem} \label{unrootedOutTreeNoKernel} \textsc{$k$-Leaf Out-Tree} has no polynomial kernel unless {\rm PH=$\Sigma_p^3$}. \end{theorem}
\begin{proof}
The problem is NP-complete~\cite{alonLNCS4596}. We prove that it is compositional and thus, Theorem \ref{thm:kernellbound} will imply the statement of the theorem. A simple composition algorithm for this problem is as follows. On input 
$(D_1, k), (D_2, k), \ldots ,(D_t, k)$ output the instance $(D,k)$ where 
$D$ is the disjoint union of $D_1, \ldots ,D_t$. Since an out-tree must be 
completely contained in a connected component of the underlying undirected graph of $D$,  $(D,k)$ is a yes instance to \textsc{$k$-Leaf Out-Tree} if and only if any out of $(D_i, k)$, $1\leq i \leq t$, is. This concludes the proof.
\end{proof}

A \emph{willow} graph~\cite{DV08}, $D = (V,A_1 \cup A_2)$ is a directed 
graph such that $D' = (V,A_1)$ is a directed path $P=p_1p_2\ldots p_n$ on all vertices of $D$ 
and $D''=(V,A_2)$ is a directed acyclic graph with one vertex $r$ of in-degree 
$0$, such that every arc of $A_2$ is a backwards arc of $P$. $p_1$ is called the \emph{bottom} vertex of the willow, $p_n$ is called the \emph{top} of the willow and $P$ is called the \emph{stem}. 
A \emph{nice willow} graph $D=(V,A_1 \cup A_2)$ is a willow graph where $p_np_{n-1}$ and $p_np_{n-2}$ are arcs of $D$, neither $p_{n-1}$ nor $p_{n-2}$ are incident to any other arcs of $A_2$ and $D''=(V,A_2)$ has a $p_n$-out-branching.

\begin{obs} $[\star]$ \label{niceWillowRoot} Let $D=(V,A_1 \cup A_2)$ be a nice willow graph. Every out-branching of $D$ with the maximum number of leaves is rooted at the top vertex $p_n$ \end{obs}

\begin{lemma} \label{nwtComplete} $[\star]$ \textsc{$k$-Leaf Out-Tree} in nice willow graphs is NP-complete under Karp reductions. \end{lemma}

\begin{theorem} \textsc{$k$-Leaf Out-Branching} has no polynomial kernel unless {\rm PH=$\Sigma_p^3$}. \end{theorem}

\begin{proof}
We prove that if \textsc{$k$-Leaf Out-Branching} has a polynomial kernel then so does \textsc{$k$-Leaf Out-Tree}. Let $(D,k)$ be an instance to \textsc{$k$-Leaf Out-Tree}. For every vertex $v \in V$ we make an instance $(D,v,k)$ to \textsc{Rooted $k$-Leaf Out-Tree}. Clearly, $(D,k)$ is a yes instance for \textsc{$k$-Leaf Out-Tree} if and only if $(D,v,k)$ is a yes instance to \textsc{Rooted $k$-Leaf Out-Tree} for some $v \in V$. By 
Theorem~\ref{thm:polykernel} \textsc{Rooted $k$-Leaf Out-Tree} has a $O(k^3)$ kernel, so we can apply the kernelization algorithm for \textsc{Rooted $k$-Leaf Out-Tree} separately on each of the $n$ instances of \textsc{Rooted $k$-Leaf Out-Tree} to get $n$ instances $(D_1,v_1,k)$, $(D_2,v_2,k)$, $\ldots$, $(D_n, v_n, k)$ with $|V(D_i)| = O(k^3)$ for each $i \leq n$. By Lemma \ref{nwtComplete} \textsc{$k$-Leaf Out-Branching} in nice willow graphs is NP-complete under Karp reductions so we can reduce each instance $(D_i,v_i,k)$ of \textsc{Rooted $k$-Leaf Out-Tree} to an instance $(W_i, b_i)$ of \textsc{$k$-Leaf Out-Branching} in nice willow graphs in polynomial time in $|D_i|$, and hence in polynomial time in $k$. Thus, in each such instance, $b_i \leq k^c$ for some fixed constant $c$ independent of both $n$ and $k$. Let $b_{max} = \max_{i \leq n} b_i$. Without loss of generality $b_i = b_{max}$ for every $i$. This assumption is safe because if it does not hold we can modify the instance $(W_i, b_i)$ by replacing $b_i$ with $b_{max}$, subdividing the last arc of the stem $b_{max}-b_i$ times and adding an edge from $r_i$ to each subdivision vertex.

From the instances $(W_1, b_{max})$, $\ldots$, $(W_n, b_{max})$ we build an instance $(D',b_{max}+1)$ to \textsc{$k$-Leaf Out-Branching}. Let $r_i$ and $s_i$ be the top and bottom vertices of $W_i$ respectively. We build $D'$ simply by taking the disjoint union of the willows graphs $W_1, W_2,\ldots, W_n$ and adding in an arc $r_is_{i+1}$ for $i < n$ and the arc $r_ns_1$. Let $C$ be the directed cycle in $D$ obtained by taking the stem of $D'$ and adding the arc $r_ns_1$.

If for any $i \leq n$, $W_i$ has an out-branching with at least $b_ {max}$ leaves, then $W_i$ has an out-branching rooted at $r_i$ with at least $b_ {max}$ leaves. We can extend this to an out-branching of $D'$ with at least $b_ {max}+1$ leaves by following $C$ from $r_i$. In the other direction suppose $D'$ has an out-branching $T$ with at least $b_ {max}+1$ leaves. Let $i$ be the integer such that the root $r$ of $T$ is in $V(W_i)$. For any vertex $v$ in $V(D')$ outside of $V(W_i)$, the only path from $r$ to $v$ in $D'$ is the directed path from $r$ to $v$ in $C$. Hence $T$ has at most $1$ leaf outside of $V(W_i)$. Thus $T[V(W_1)]$ contains an out-tree with at least $b_ {max}$ leaves.

By assumption, \textsc{$k$-Leaf Out-Branching} has a polynomial kernel. Hence we can apply a kernelization algorithm to get an instance $(D'',k'')$ of \textsc{$k$-Leaf Out-Branching} with $|V(D'')| \leq b_{max}^{c_2}$ for a constant $c_2$ independent of $n$ and $b_{max}$ such that $(D'',k'')$ is a yes instance if and only if $(D',b_{max})$ is.

Finally, since \textsc{$k$-Leaf Out-Tree} is NP-complete we can reduce $(D'',k'')$ to an instance $(D^*,k^*)$ of \textsc{$k$-Leaf Out-Tree} in polynomial time. Hence $k^* \leq |V(D^*)| \leq |V(D'')|^{c_3} \leq k^{c_4}$ for some fixed constants $c_3$ and $c_4$. Hence we conclude that if \textsc{$k$-Leaf Out-Branching} has a polynomial kernel then so does \textsc{$k$-Leaf Out-Tree}. Thus, Theorem \ref{unrootedOutTreeNoKernel} implies that \textsc{$k$-Leaf Out-Branching} has no polynomial kernel unless {\rm PH=$\Sigma_p^3$}.
\end{proof}

\section{Conclusion and Discussions}
In this paper we demonstrate that Turing kernelization is a more poweful technique than many-to-one kernelization.
We showed that while {\sc $k$-Leaf Out-Branching} and {\sc $k$-Leaf Out-Tree} do not have a polynomial kernel 
unless an unlikely collapse of complexity classes occurs, they do have $n$ independent cubic kernels. Our paper 
raises far more questions than it answers. We believe that there are many more problems waiting to be addressed 
from the viewpoint of Turing kernelization.  A few concrete open problems in this direction are as follows.
\begin{itemize}
\setlength{\itemsep}{-2pt}
\item Is there a framework to rule out the possibility of $|I|^{O(1)}$ polynomial kernels similar to the framework developed in~\cite{BDFH08}?
\item Which other problems admit a Turing kernelization like the cubic kernels for {\sc $k$-Leaf Out-Branching} and {\sc $k$-Leaf Out-Tree} obtained here?
\item Does there exist a problem for which we do not have a linear many-to-one kernel, but does have linear kernels from the viewpoint of Turing kernelization?
\end{itemize}


\newpage

\begin{appendix}

\section{Proofs moved from Section~\ref{section:redrule}}
\label{appendix:redrule}
\setcounter{section}{3}
\renewcommand{\thesection}{\arabic{section}}

\setcounter{lemma}{0}
\begin{lemma}
Reduction Rules \ref{reachRule} and \ref{uslessArcRule} are safe.
\end{lemma}
\begin{proof} If there exists a vertex which can not be reached from the root $r$ then a digraph can not have
any $r$-out-branching.
For Reduction Rule \ref{uslessArcRule},
all paths from $r$ to $v$ contain the vertex $u$ and thus the arc $vu$ is a back arc in any $r$-out-branching of $D$.
\end{proof}

\setcounter{lemma}{1}

\begin{lemma}
 Reduction Rule \ref{bridgeRule} is safe.
\end{lemma}

\begin{proof}  Let the arc $uv$ disconnect at least two vertices $v$ and $w$ from $r$ and let $D'$ be the digraph obtained from $D$ by contracting the arc $uv$. Let $T$ be an $r$-out-branching of $D$ with at least $k$ leaves. Since every path from $r$ to $w$ contains the arc $uv$, $T$ contains $uv$ as well and neither $u$ nor $v$ are leaves of $T$. Let $T'$ be the tree obtained from $T$ by contracting $uv$. $T'$ is an $r$-out-branching of $D'$ with at least $k$ leaves.

In the opposite direction, let $T'$ be an $r$-out-branching of $D'$ with at least $k$ leaves. Let $u'$ be the vertex in $D'$ obtained by contracting the arc $uv$, and let $x$ be the parent of $u'$ in $T'$. Notice that the arc $xu'$ in $T'$ was initially the arc $xu$ before the contraction of $uv$,
since there is no path from $r$ to $v$ avoiding $u$ in $D$. We make an $r$-out-branching $T$ of $D$ from $T'$,
by replacing the vertex $u'$ by the vertices $u$ and $v$ and adding the arcs $xu$, $uv$ and arc sets
$\{vy : u'y \in A(T') \wedge vy \in A(D)\}$ and
$\{uy : u'y \in A(T') \wedge vy \notin A(D)\}$.
All these arcs belong to $A(D)$ because all out-neighbors of $u'$ in $D'$
are out-neighbors either of $u$ or of $v$ in $D$.
Finally $u'$ must be an inner vertex of $T'$ since $u'$ disconnects $w$ from $r$.
Hence $T$ has at least as many leaves as $T'$.
\end{proof}

\setcounter{lemma}{2}
\begin{lemma}
Reduction Rule \ref{AvoidableArcRule} is safe.
\end{lemma}
\begin{proof}
Let $D'$ be the graph obtained by removing the arc $vw$ from $D$ and let $T$ be an $r$-out-branching of $D$.
If $vw \notin A(T)$, $T$ is an $r$-out-branching of $D'$, so suppose $vw \in A(T)$.
Any $r$-out-branching of $D$ contains the vertex $v$,
and since all paths from $r$ to $v$ contain some vertex $x \in S$,
some vertex $u \in S$ is an ancestor of $v$ in $T$.
Let $T' = (T \cup uw) \setminus vw$. $T'$ is an out-branching of $D'$.
Furthermore, since $u$ is an ancestor of $v$ in $T$, $T'$ has at least as many leaves as $T$.
For the opposite direction observe that any $r$-out-branching of $D'$ is also an $r$-out-branching of $D$.
\end{proof}

\setcounter{lemma}{3}

\begin{lemma}
Reduction Rule \ref{twoDirectionalPathRule} is safe.
\end{lemma}
\begin{proof}
Let $D'$ be the graph obtained by performing Reduction Rule \ref{twoDirectionalPathRule} to a path $P$ in $D$. Let $P_u$ be the path $p_1p_{out}uvp_{in}p_l$ and $Q_v$ be the path $p_{in}p_lvup_1p_{out}$. Notice that $P_u$ is the unique out-branching of $D'[V(P_u)]$ rooted at $p_1$ and that $Q_v$ is the unique out-branching of $D'[V(P_u)]$ rooted at $p_{in}$.

Let $T$ be an $r$-out-branching of $D$ with at least $k$ leaves. Notice that since $P$ is the unique out-branching of $D[V(P)]$ rooted at $p_1$, $Q$ is the unique out-branching of $D[V(P)]$ rooted at $p_{in}$ and $p_1$ and $p_{in}$ are the only vertices with in-arcs from the outside of $P$, $T[V(P)]$ is either a path or the union of two vertex disjoint paths. Thus, $T$ has at most two leaves in $V(P)$ and at least one of the following three cases must apply.

\begin{enumerate}
\setlength{\itemsep}{-2pt}
 \item $T[V(P)]$ is the path $P$ from $p_1$ to $p_l$.
 \item $T[V(P)]$ is the path $Q$ from $p_{in}$ to $p_{out}$.
 \item $T[V(P)]$ is the vertex disjoint union of a path $\tilde{P}$ that is a subpath of $P$ rooted at $p_1$, and a path $\tilde{Q}$ that is a subpath of $Q$ rooted at $p_{in}$.
\end{enumerate}

In the first case we can replace the path $P$ in $T$ by the path $P_u$ to get an $r$-out-branching of $D'$ with at least $k$ leaves. Similarly, in the second case, we can replace the path $Q$ in $T$ by the path $Q_v$ to get an $r$-out-branching of $D'$ with at least $k$ leaves. For the third case, observe that $\tilde{P}$ must contain $p_{out}$ since $p_{out} = p_1$ or $p_1$ appears before $p_{out}$ on $Q$ and thus, $p_{out}$ can only be reached from $p_1$. Similarly, $\tilde{Q}$ must contain $p_l$. Thus, $T \setminus R$ is an $r$-out-branching of $D \setminus R$. We build an $r$-out-branching $T'$ of $D'$ by taking $T \setminus R$ and letting $u$ be the child of $p_{out}$ and $v$ be the child of $p_l$. In this case $T$ and $T'$ have same number of leaves outside of $V(P)$ and $T$ has at most two leaves in $V(P)$ while both $u$ and $v$ are leaves in $T'$. Hence $T'$ has at least $k$ leaves.

In the other direction let $T'$ be an $r$-out-branching of $D'$ with at least $k$ leaves. Notice that since $P_u$ is the unique out-branching of $D'[V(P_u)]$ rooted at $p_1$, $Q_v$ is the unique out-branching of $D'[V(P_u)]$ rooted at $p_{in}$ and $p_1$ and $p_{in}$ are the only vertices with in-arcs from the outside of $V(P_u)$, $T'[V(P_u)]$ is either a path or the union of two vertex disjoint paths. Thus, $T'$ has at most two leaves in $V(P_u)$ and at least one of the following three cases must apply.

\begin{enumerate}
\setlength{\itemsep}{-2pt}
 \item $T'[V(P_u)]$ is the path $P_u$ from $p_1$ to $p_l$.
 \item $T'[V(P_u)]$ is the path $Q_v$ from $p_{in}$ to $p_{out}$.
 \item $T'[V(P_u)]$ is the vertex disjoint union of a path $\tilde{P_u}$ that is a subpath of $P_u$ rooted at $p_1$, and a path $\tilde{Q_v}$ that is a subpath of $Q_v$ rooted at $p_{in}$.
\end{enumerate}

In the first case we can replace the path $P_u$ in $T'$ by the path $P$ to get an $r$-out-branching of $D$ with at least $k$ leaves. Similarly, in the second case, we can replace the path $Q_v$ in $T'$ by the path $Q$ to get an $r$-out-branching of $D'$ with at least $k$ leaves. For the third case, observe that $\tilde{P_u}$ must contain $p_{out}$ since $p_{out} = p_1$ or $p_1$ appears before $p_{out}$ on $Q_v$ and thus, $p_{out}$ can only be reached from $p_1$. Similarly, $\tilde{Q_v}$ must contain $p_l$.
Thus, $T' \setminus \{u,v\}$ is an $r$-out-branching of $D' \setminus \{u,v\}$. Let $x$ be the vertex after $p_{out}$ on $P$, and let $y$ be the vertex after $p_l$ on $Q$.
Vertices $x$ and $y$ must be distinct vertices in $R$ and thus there must be two vertex disjoint paths $P_x$ and $Q_y$ rooted at $x$ and $y$ respectively so that $V(P_x) \cup V(Q_y) = R$. We build an $r$-out-branching $T$ from $(T' \setminus \{u,v\} ) \cup P_x \cup Q_y$ by letting $x$ be the child of $p_{out}$ and $y$ be the child of $p_{in}$. In this case $T'$ and $T$ have the same number of leaves outside of $V(P)$ and $T'$ has at most two leaves in $V(P_u)$ while both the leaf of $P_u$ and the leaf of $Q_v$ are leaves in $T$. Hence $T$ has at least $k$ leaves.
%
\end{proof}

\setcounter{section}{1}
\renewcommand{\thesection}{\Alph{section}}

\section{Proofs moved from Section~\ref{section:polykernelno}}
\label{appendix:polykernelno}

\setcounter{section}{4}
\renewcommand{\thesection}{\arabic{section}}

\setcounter{lemma}{4}
\begin{lemma}
There is a nice forest in $P$.
\end{lemma}

\begin{proof}
We define a subgraph $F$ of $D[V(P)]$ as follows.
The vertex set of $F$ is $V(P)$ and an arc $p_tp_s$ is in $A(F)$ if $p_s \not\in S$
and $t$ is the largest number so that $p_tp_s \in A(D)$. Notice that all arcs 
of $F$ are covered by property $(b)$ in the definition of a nice forest.  

We prove that $F$ is a forest.
Suppose for contradiction that there is a cycle $C$ in $F$.
By definition of $F$ every vertex has in-degree at most $1$,
$C$ must be a directed cycle.
Since every vertex in $S$ has in-degree $0$ in $F$,
$C \cap S = \emptyset$.
Consider the highest numbered vertex $p_i$ on $C$.
Since $P$ has no forward arcs, $p_{i-1}$ is the predecessor of $p_i$ in $C$.
The construction of $F$ implies that there can not be an arc $p_qp_i$ where $q > i$ in $A(D)$.
Also, $p_i$ does not have any in-arcs from outside of $P$.
Thus, $p_{i-1}$ disconnects $p_i$ from the root.
Hence, by Rule \ref{uslessArcRule} $p_{i}p_{i-1} \not\in A(D)$.
Let $p_j$ be the predecessor of $p_{i-1}$ in $C$.
Then $j < i-1$, since $p_{i}p_{i-1} \not\in A(D)$ and $p_i$ is the highest numbered vertex in $C$.
Hence $j = i-2$.
This contradicts that $D$ is a reduced instance since the arc $p_{i-2}p_{i-1}$
disconnects $p_{i-1}$ and $p_i$ from the root $r$ implying that Rule \ref{bridgeRule} can be applied.
Since $F$ is a forest and since every vertex in $V(P)$ except for vertices in $S$ have in-degree $1$
we conclude that $F$ is a forest of directed trees rooted at vertices in $S$.
Since $F$ is a forest and $P$ has no forward arcs, $F$ is a nice forest.
\end{proof}
\setcounter{lemma}{5}

\begin{lemma} 
\label{le:keyVertexCount} Let $F$ be a nice forest of $P$. There are at most $5(k-1)$ key vertices of $F$. \end{lemma}

\begin{proof}
By the proof of Corollary \ref{cor:fewInNodes} there is an $r$-out-tree $T_S$
with $(V(T_S) \cap V(P)) \subseteq (S \cup \{p_l\})$ and $(A(T_S) \cap A(P)) = \emptyset$,
such that all vertices in $S \setminus \{p_l\}$ are leaves of $T_S$.
We build an $r$-out-tree $T_F = (V(T_S) \cup V(P), A(T_S) \cup A(F))$.
Notice that every leaf of $F$ is a leaf of $T_F$, except possibly for $p_l$.
Since $D$ is a no-instance $T_F$ has at most $k-1$ leaves
and $k-2$ vertices with out-degree at least $2$.
Thus, $F$ has at most $k$ leaves and at most $k-2$ vertices with out-degree at least $2$.
Hence the number of vertices in $F$
whose parent in $F$ has out-degree at least $2$ is at most $2k-2$.
Finally, by Corollary \ref{cor:fewInNodes}, $|S| \leq k$.
Adding up these upper bounds yields that there are at most $k-1+k-2+2k-2+k=5(k-1)$ key vertices of $F$.
\end{proof}

\setcounter{lemma}{7}
\begin{corollary}\label{cor:noForward}
No arc $p_{i}p_{i+1}$ is a forward arc of $F$.
\end{corollary}

\begin{proof}
If $p_{i}p_{i+1}$ is a forward arc of $F$ then there is a path from $p_{i}$ to $p_{i+1}$ in $F$. By Observation \ref{obs:uniquePath} $p_ip_{i+1}$ is the unique path from $p_i$ to $p_{i+1}$ in $D[V(P)]$. Hence $p_ip_{i+1} \in A(F)$ contradicting that it is a forward arc. \end{proof}

\setcounter{lemma}{9}
\begin{obs} \label{obs:noTwoFront}
If neither $p_i$ nor $p_{i+1}$ are key vertices,
then either $p_ip_{i+1} \notin A(F)$ or $p_{i+1}p_{i+2} \notin A(F)$.
\end{obs}
\begin{proof} Assume for contradiction that $p_ip_{i+1} \in A(F)$ and  $p_{i+1}p_{i+2} \in A(F)$. Since neither $p_i$ nor $p_{i+1}$ are key vertices, both $p_{i+1}$ and $p_{i+2}$ must have in-degree $1$ in $D$. Then the arc $p_ip_{i+1}$ disconnects both $p_{i+1}$ and $p_{i+2}$ from the root $r$ and Rule \ref{bridgeRule} can be applied, contradicting that $D$ is a reduced instance. \end{proof}

\setcounter{lemma}{10}
\begin{lemma}\label{le:FInducesPath}
$V(P')$ induces a directed path in $F$.
\end{lemma}
\begin{proof}
We first prove that for any arc $p_ip_{i+1} \in A(P')$ such that $p_ip_{i+1} \notin A(F)$, there is a path from $p_{i+1}$ to $p_i$ in $F$. Suppose for contradiction that there is no path from $p_{i+1}$ to $p_i$ in $F$, and
let $x$ be the parent of $p_{i+1}$ in $F$.
Then $p_ip_{i+1}$ is not a backward arc of $F$ and hence $F'= (F \setminus xp_{i+1}) \cup \{p_ip_{i+1}\}$ is a forest of out-trees rooted at vertices in $S$. Also, since $p_{i+1}$ is not a key vertex, $x$ has out-degree $1$ in $F$ and thus $x$ is a leaf in $F'$. Since $p_i$ is not a leaf in $F$, $F'$ has one more leaf than $F$. Now, every vertex with out-degree at least $2$ in $F$ has out-degree at least $2$ in $F'$. Additionally, $p_i$ has out-degree $2$ in $F'$. Hence $F'$ is a nice forest of $P$ with more leaves than $F$, contradicting the choice of $F$.

Now, notice that by Observation \ref{obs:uniquePath} any path in $D[V(P)]$ from a vertex $u \in V(P')$ to a vertex $v \in V(P')$ that contains a vertex $w \notin V(P')$ must contain either the arc $p_{x-1}p_x$ or the arc $p_{y}p_{y+1}$. 
Since neither of those two arcs are arcs of $F$ it follows that for any arc $p_ip_{i+1} \in A(P')$ such that $p_ip_{i+1} \notin A(F)$, there is a path from $p_{i+1}$ to $p_i$ in $F[V(P')]$.
Hence $F[V(P')]$ is weakly connected, that is, the underlying undirected graph is connected.
Since every vertex in $V(P')$ has in-degree $1$ and out-degree $1$ in
$F$ we conclude that $F[V(P')]$ is a directed path.
\end{proof}

\setcounter{lemma}{11}
\begin{obs}\label{obs:PFOrder}
For any pair of vertices $p_i,p_j \in V(P')$ if $i \leq j-2$
then $p_j$ appears before $p_i$ in $Q'$.
\end{obs}
\begin{proof}
Suppose for contradiction that $p_i$ appears before $p_j$ in $Q'$.
By Observation \ref{obs:uniquePath} $p_ip_{i+1}p_{i+2}\ldots p_j$
is the unique path from $p_i$ to $p_j$ in $D[V(P')]$.
This path contains both the arc $p_ip_{i+1}$ and $p_{i+1}p_{i+2}$
contradicting Observation \ref{obs:noTwoFront}.
\end{proof}

\setcounter{lemma}{12}
\begin{lemma}\label{le:PAndFIsAll}
All arcs of $D[V(P')]$ are contained in $A(P') \cup A(F)$.
\end{lemma}
\begin{proof}
Since $P$ has no forward arcs it is enough to prove that any
arc $p_jp_i \in A(D[V(P')])$ with $i < j$ is an arc of $F$.
Suppose this is not the case and let $p_q$ be the parent of $p_i$ in $F$.
We know that $p_i$ has in-degree at least $2$ in $D$ and
also since $p_i$ is not a key vertex $p_q$ has in-degree one in $F$.
Hence by definition of $F$ being a nice forest, we have that for every $t > q$, $p_tp_i \notin A(D)$.
It follows that $i < j < q$.
By Lemma \ref{le:FInducesPath} $F[V(P')]$ is a directed path $Q'$ containing both $p_i$ and $p_j$.
If $p_j$ appears after $p_i$ in $Q'$,
Observation \ref{obs:PFOrder} implies that $i= j-1$ and that $p_j$
has in-degree $1$ in $D$ since $F$ is a nice forest.
Thus $p_i$ separates $p_j$ from the root and Rule \ref{uslessArcRule} can be applied to $p_jp_i$
contradicting that $D$ is a reduced instance.
Hence $p_j$ appears before $p_i$ in $Q'$.

Since $p_j$ is an ancestor of $p_i$ in $F$ and $p_q$ is the parent of $p_i$ in $F$,
$p_j$ is an ancestor of $p_q$ in $F$ and hence $p_q \in V(Q') = V(P')$.
Now, $p_j$ comes before $p_q$ in $Q'$ and $j < q$ so
Observation \ref{obs:PFOrder} implies that $q = j+1$ and
that $p_q$ has in-degree $1$ in $D$ since $F$ is a nice forest.
Thus $p_j$ separates $p_q$ from the root $r$ and both $p_jp_i$ and $p_qp_i$ are arcs of $D$.
Hence Rule \ref{AvoidableArcRule} can be applied to remove the arc $p_qp_i$ contradicting that
$D$ is a reduced instance.
\end{proof}

\setcounter{lemma}{14}

\begin{obs} \label{obs:uniqueBackPath}
Let $Q' = F[V(P')]$.
For any pair of vertices $u,v$ such that there is a path $Q'[uv]$ from $u$ to $v$ in $Q'$,
$Q'[uv]$ is the unique path from $u$ to $v$ in $D[V(P')]$.
\end{obs}
\begin{proof}
By Lemma \ref{le:FInducesPath} $Q'$ is a directed path $f_1f_2\ldots f_{|P'|}$ and
let $Q'[f_1f_i]$ be the path $f_1f_2\ldots f_i$.
We prove that for any $i < |Q'|$, $f_if_{i+1}$ is the only arc
from $V(Q'[f_1f_i])$ to $V(Q'[f_{i+1}f_{|P'|}])$.
By Lemma \ref{le:PAndFIsAll} all arcs of $D[V(P')]$ are either arcs of $P'$ or arcs of $Q'$.
Since $Q'$ is a path, $f_if_{i+1}$ is the only arc from $V(Q'[f_1f_i])$ to $V(Q'[f_{i+1}f_{|P'|}])$ in $Q'$.
By Corollary \ref{cor:noForward} there are no arcs from $V(Q'[f_1f_i])$ to
$V(Q'[f_{i+1}f_{|P'|}])$ in $P'$, except possibly for $f_if_{i+1}$.
\end{proof}

\setcounter{lemma}{16}
\begin{corollary} \label{cor:fewArcsOut}
There are at most $14(k-1)$ vertices in $P'$ with arcs to vertices outside of $P'$.
\end{corollary}
\begin{proof}
By Lemma \ref{le:fewOutVertices} there are at most $7(k-1)$ vertices that are endpoints of arcs originating in $P'$.
By Lemma \ref{le:twoArcsOut} each such vertex is the endpoint of at most $2$ arcs from vertices in $P'$.
\end{proof}

\setcounter{lemma}{18}
\begin{lemma} 
\label{lemma:boundpath}
Let $D$ be a reduced no-instance to \textsc{Rooted $k$-Leaf Out-Branching}. 
Then $|V(D)| = O(k^3)$. 
\end{lemma}
\begin{proof}
Let $T$ be a BFS-tree of $D$. $T$ has at most $k-1$ leaves and at most $k-2$ inner vertices with out-degree at least $2$. The remaining vertices can be partitioned into at most $2k-3$ paths $P_1 \ldots P_t$ with all vertices having out-degree $1$ in $T$.
We prove that for every $q \in \{1, \ldots, t\}$, $|P_q| = O(k^2)$. Let $F$ be a nice forest of $P_q$ with the maximum number of leaves. By Lemma \ref{le:keyVertexCount}, $F$ has at most $5(k-1)$ key vertices. Let $p_i$ and $p_j$ be consecutive key vertices of $F$ on $P_q$. By Observation \ref{obs:noTwoFront}, there is a path $P' = p_xp_{x+1} \ldots p_y$ containing no key vertices, with $x \leq i+1$ and $y \geq j-1$, such that neither $p_{x-1}p_x$ nor $p_yp_{y+1}$ are arcs of $F$. 
By Lemma \ref{le:boundPPrime} 
$|P'| \leq 154(k-1)+10$ so 
$|P_q| \leq (5(k-1)+1)(154(k-1)+10) + 3(5(k-1))$. 
Hence, $|V(D)| \leq 2k(5k(154(k-1)+10+3)) \leq 1540k^3= O(k^3)$.
\end{proof}

\setcounter{section}{2}
\renewcommand{\thesection}{\Alph{section}}

\section{Proofs moved from Section~\ref{cheatKernel}}
\setcounter{section}{5}
\renewcommand{\thesection}{\arabic{section}}

\setcounter{lemma}{3}
\begin{obs} 
Let $D=(V,A_1 \cup A_2)$ be a nice willow graph. Every out-branching of $D$ with the maximum number of leaves is rooted at the top vertex $p_n$ \end{obs}
\begin{proof}
Let $P=p_1p_2\ldots p_n$ be the stem of $D$ and suppose for contradiction that there is an out-branching $T$ with the maximum number of leaves rooted at $p_i$, $i < n$. Since $D$ is a nice willow $D'=(V,A_2)$ has a $p_n$-out-branching $T'$. Since every arc of $A_2$ is a back arc of $P$, $T'[\{v_j : j \geq i\}]$ is an $p_n$-out-branching of $D[\{v_j : j \geq i\}]$. Then $T'' = (V,\{v_xv_y \in A(T'): y \geq i\} \cup \{v_xv_y \in A(T): y < i\})$ is an out-branching of $D$.
If $i = n-1$ then $p_n$ is not a leaf of $T$ since the only arcs going out of the set $\{p_n, p_{n-1}\}$ start in $p_n$. Thus, in this case, all leaves of $T$ are leaves of $T''$ and $p_{n-1}$ is a leaf of $T''$ and not a leaf of $T$, contradicting that $T$ has the maximum number of leaves.
\end{proof}

\setcounter{lemma}{4}

\begin{lemma} 
\textsc{$k$-Leaf Out-Tree} in nice willow graphs is NP-complete under Karp reductions. \end{lemma}
\begin{proof}
We reduce from the well known NP-complete \textsc{Set Cover} problem  \cite{Karp72}. A \emph{set cover} of a universe $U$ is a family ${\cal F}'$ of sets over $U$ such that every element of $u$ appears in some set in ${\cal F}'$. In the \textsc{Set Cover} problem one is given a family ${\cal F} = \{S_1, S_2, \ldots S_m\}$ of sets over a universe $U$, $|U|=n$, together with a number $b \leq m$ and asked whether there is a set cover ${\cal F}' \subset {\cal F}$ with $|{\cal F}'| \leq b$ of $U$. In our reduction we will assume that every element of $U$ is contained in at least one set in ${\cal F}$. We will also assume that $b \leq m-2$.
These assumptions are safe because if either of them does not hold, the \textsc{Set Cover} instance can be resolved in polynomial time. From an instance of \textsc{Set Cover} we build a digraph $D=(V,A_1 \cup A_2)$ as follows. The vertex set $V$ of $D$ is a root $r$, vertices $s_i$ for each $1 \leq i \leq m$ representing the sets in ${\cal F}$, vertices $e_i$, $1 \leq i \leq n$ representing elements in $U$ and finally $2$ vertices $p$ and $p'$.

The arc set $A_2$ is as follows, there is an arc from $r$ to each vertex $s_i$, $1 \leq i \leq m$ and there is an arc from a vertex $s_i$ representing a set to a vertex $e_j$ representing an element if $e_j \in S_i$. Furthermore, $rp$ and $rp'$ are arcs in $A_2$. Finally, we let $A_1 = \{e_{i+1}e_i : 1 \leq i < n\} \cup \{s_{i+1}s_i : 1 \leq i < m\} \cup \{e_1s_m, s_1p, pp', p'r\}$. This concludes the description of $D$. We now proceed to prove that there is a set cover ${\cal F}' \subset {\cal F}$ with $|{\cal F}'| \leq b$ if and only if there is an out-branching in $D$ with at least $n+m+2-b$ leaves.

Suppose that there is a set cover ${\cal F}' \subset {\cal F}$ with $|{\cal F}'| \leq b$. We build a directed tree $T$ rooted at $r$ as follows. Every vertex $s_i$, $1 \leq i \leq m$, $p$ and $p'$ has $r$ as their parent. For every element $e_j$, $1 \leq i \leq n$ we chose the parent of $e_j$ to be $s_i$ such that $e_j \in S_i$ and $S_i \in {\cal F}'$ and for every $i' < i$ either $S_{i'} \notin |{\cal F}'|$ or $e_j \notin S_{i'}$. Since the only inner nodes of $T$ except for the root $r$ are vertices representing sets in the set cover, $T$ is an out-branching of $D$ with at least $n+m+2-b$ leaves.

In the other direction suppose that there is an out-branching $T$ of $D$ with at least $n+m+2-b$ leaves, and suppose that $T$ has the most leaves out of all out-branchings of $D$. Since $D$ is a nice willow with $r$ as top vertex, Observation \ref{niceWillowRoot} implies that $T$ is an $r$-out-branching of $D$. Now, if there is an arc $e_{i+1}e_i \in A(T)$ then let $s_j$ be a vertex such that $e_i \in S_j$. Then $T' = (T \setminus e_{i+1}e_i) \cup s_je_i$ is an $r$-out-branching of $D$ with as many leaves as $T$. Hence, without loss of generality, for every $i$ between $1$ and $n$, the parent of $e_i$ in $T$ is some $s_j$. Let ${\cal F'} = \{S_i : s_i$ is an inner vertex of $T\}$. ${\cal F'}$ is a set cover of $U$ with size at most $n+m+2-(n+m+2-b) = b$, concluding the proof.
\end{proof}

\end{appendix}

\end{document}